\documentclass{article}

\usepackage{scrextend}
\usepackage{fullpage}
\usepackage{amssymb}
\usepackage{color}
\usepackage{amsthm}
\usepackage{microtype}
\usepackage{hyperref}
\usepackage[cmex10]{amsmath}

\usepackage[font={small}]{caption}

\newtheorem{theorem}{Theorem}
\newtheorem{lemma}{Lemma}
\newtheorem{example}{Example}

\newtheorem{corollary}{Corollary}
\newtheorem{definition}{Definition}

\newcommand{\size}[1]{\mathsf{size}(#1)}
 
\usepackage{microtype}
\usepackage{graphicx}

\newtheorem*{claim}{Claim}

\newcommand{\var}{\mathsf{var}}

\def\hy{\hbox{-}\nobreak\hskip0pt} 
\newcommand{\SB}{\{\,}  \newcommand{\SE}{\,\}}

\begin{document}

\title{On Compiling Structured CNFs to OBDDs}
\author{Simone Bova\thanks{Supported by the European Research Council (Complex Reason, 239962) 
and the FWF Austrian Science Fund (Parameterized Compilation, P26200).} \and Friedrich Slivovsky
\thanks{Supported by the European Research Council (Complex Reason, 239962).}}

\date{\small Vienna University of Technology}

\maketitle

\begin{abstract}
We present new results on the size of OBDD 
  representations of structurally characterized classes of CNF formulas.  
First, we identify a natural sufficient condition, which we call the \emph{few subterms} property, 
for a class of CNFs to have polynomial OBDD size; 
we then prove that CNFs whose incidence graphs are \emph{variable convex} 
have few subterms (and hence have polynomial OBDD size), and observe that 
the few subterms property also explains the known fact that classes of CNFs of \emph{bounded treewidth} 
have polynomial OBDD size.  Second, we prove an exponential lower bound on the OBDD size of a 
family of CNF classes with incidence graphs of \emph{bounded degree}, 
exploiting the combinatorial properties of \emph{expander graphs}.
\end{abstract}

\section{Introduction}

\noindent \textit{Motivation.} A fundamental theoretical task in the study of Boolean functions is to relate the size of their encodings in different representation languages. In particular, the representation of circuits as binary decision diagrams (also known as branching programs) has been the subject of intense study in complexity theory (see, for instance~\cite[Chapter~14]{W00} and \cite[Part V]{J12}).  
In this paper, we study the \emph{ordered binary decision diagram (OBDD)} representations of Boolean functions given as propositional formulas in \emph{conjunctive normal form (CNF)}. In contrast to other variants of binary decision diagrams, equivalence of OBDDs can be decided in polynomial time, a crucial feature for basic applications in the areas of verification and synthesis \cite{BW99}.  \begin{figure}
\begin{center}
\input{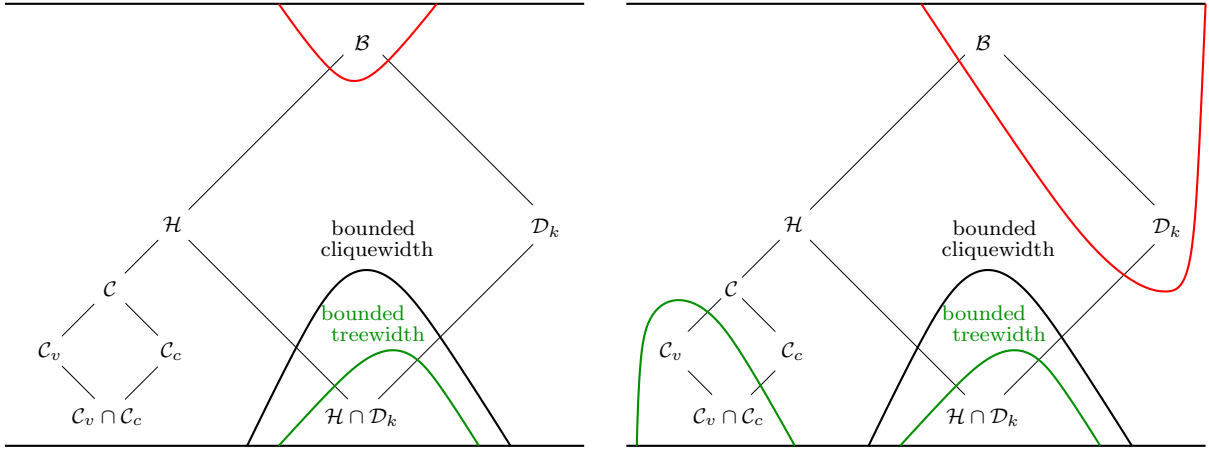} 
\end{center}
\caption{The diagram depicts a hierarchy of classes of bipartite graphs under the inclusion relation (thin edges).  
$\mathcal{B}$, $\mathcal{H}$, $\mathcal{D}_k$, $\mathcal{C}$, $\mathcal{C}_v$, and $\mathcal{C}_c$ 
denote, respectively, bipartite graphs, chordal bipartite graphs (corresponding to beta acyclic CNFs), 
bipartite graphs of degree at most $k$ ($k \geq 3$), 
convex graphs, left (variable) convex graphs, and right (clause) convex graphs.  
The class $\mathcal{C}_v \cap \mathcal{C}_c$ of biconvex graphs 
and the class $\mathcal{D}_k$ of bipartite graphs of degree at most $k$ 
have unbounded clique\hy width.  The class $\mathcal{H} \cap \mathcal{D}_k$ of 
chordal bipartite graph of degree at most $k$ has bounded treewidth.  The green and red 
curved lines enclose, respectively, classes of incidence graphs whose CNFs have 
polynomial time OBDD compilation, and classes of incidence graphs whose CNFs 
have exponential size OBDD representations; the right hand picture 
shows the compilability frontier, updated in light of Results 1 and 2.}
\label{fig:hierarchy}
\end{figure} 

Perhaps somewhat surprisingly, the question of which classes of CNFs can be represented as (or \emph{compiled} into, in the jargon of knowledge representation) OBDDs of polynomial size is largely unexplored~\cite[Chapter~4]{W00}. We approach this classification problem by considering \emph{structurally} characterized CNF classes, that is, 
classes of CNF formulas defined in terms of properties of their \emph{incidence graphs} (the incidence graph of a formula is the bipartite graph on clauses and variables where a variable is adjacent to the clauses it occurs in). Figure~\ref{fig:hierarchy} depicts a hierarchy of well\hy studied bipartite graph classes as considered by Lozin and Rautenbach~\cite[Figure~2]{LR04}. This hierarchy is particularly well\hy suited for our classification project as it includes prominent cases such as beta acyclic CNFs~\cite{BBCM14} and bounded clique\hy width CNFs \cite{STV14}. When located within this hierarchy, the known bounds on the OBDD size of structural CNF classes leave a large gap (depicted \emph{on the left} of Figure~\ref{fig:hierarchy}):
\begin{itemize}
\item On the one hand, we have a polynomial upper bound on the OBDD size of bounded treewidth CNF classes proved recently by Razgon \cite{R14}. The corresponding graph classes are located at the bottom of the hierarchy.
\item On the other hand, there is an exponential lower bound for the OBDD size of general CNFs, proved two decades ago by Devadas~\cite{D93}. The corresponding graph class is not chordal bipartite, has unbounded degree and unbounded clique\hy width, and hence is located at the top of the hierarchy.
\end{itemize}
\textit{Contribution.} In this paper, we tighten this gap as illustrated \emph{on the right} in Figure~\ref{fig:hierarchy}.  More specifically, we prove new bounds for two structural classes of CNFs. 

\medskip \noindent Our first result is a polynomial upper bound:

\medskip \noindent \textbf{Result 1.} CNFs whose incidence graphs are \emph{variable convex} have polynomial OBDD size (Theorem~\ref{th:leftconvex}).

\medskip Convexity is a property of bipartite graphs that has been extensively studied in the area of combinatorial optimization~\cite{Glover67,Gallo84,SY93}, 
and that can be detected in linear time~\cite{BL76,KKLV11}. To prove Result~1, we define a property of CNF classes, called the \emph{few subterms} property, that naturally arises as a sufficient condition for
polynomial size compilability when considering OBDD representations of CNF formulas (Theorem~\ref{th:psizeimpliesptime}), and then prove that CNFs with variable convex incidence graphs have this property (Lemma~\ref{lemma:convexsmallproj}). The few subterms property can also be invoked in proving the previously known result that classes of CNFs with incidence graphs of \emph{bounded treewidth} have OBDD representations of polynomial size (Lemma~\ref{lemma:twfewsubterms}). In fact, both the result on variable convex CNFs and the result on bounded treewidth CNFs can be improved to polynomial \emph{time} compilation by appealing to a stronger version of the few subterms property (Theorem~\ref{th:leftconvex} and Theorem~\ref{th:twpolytime}).

In an attempt to push the few subterms property further, we adopt the language of \emph{parameterized complexity} to formally capture the idea that CNFs \lq\lq close\rq\rq\ to a class with few subterms have \lq\lq small\rq\rq\ OBDD representations. More precisely, defining the \emph{deletion distance} of a CNF from a CNF class 
as the number of its variables or clauses that have to be deleted in order for the resulting formula to be in the class, we prove that CNFs have fixed\hy parameter tractable OBDD size parameterized by the deletion distance from a CNF class with few subterms (Theorem~\ref{th:fptsizecompiln}).  
This result can again be improved to fixed\hy parameter \emph{time} compilation under additional assumptions (Theorem~\ref{th:fptcompil}), yielding for instance fixed\hy parameter tractable time compilation of CNFs into OBDDs parameterized by the \emph{feedback vertex set} size (Corollary~\ref{cor:fvs}).

\medskip \noindent Our second result is an exponential lower bound:

\medskip \noindent \textbf{Result 2.} There is a class of CNF formulas with incidence graphs of bounded degree such that every formula~$F$ in this class has OBDD size at least $2^{\Omega(\mathsf{size}(F))}$, where $\mathsf{size}(F)$ denotes the number of variable occurrences in $F$ (Theorem~\ref{thm:mainnontechnical}).

\medskip Observe that this bound is tight: every CNF on $n$ variables has an OBDD of size $O(2^n)$. 
To establish the lower bound we use the powerful combinatorial machinery of \emph{expander graphs}. 
Despite expander graphs appearing in many areas of mathematics and computer science \cite{HLW06,L11}, 
including circuit and proof complexity~\cite{J12}, their application in this setting is novel and allows us to 
improve the best known lower bound on the OBDD size of CNFs~\cite{D93} in two ways. 
\begin{itemize}
 \item First, the formulas 
used to prove the latter bound give rise to OBDDs of size $2^{\Omega(n)}$ but ``only'' yield lower bounds of the form 
$2^{\Omega(\sqrt{\mathsf{size}(F)})}$. 
\item Second, our lower bound is established for CNF formulas that satisfy strong syntactic restrictions: 
each clause has exactly two positive literals and each variable occurs at most $3$ times; 
in particular, it holds for \emph{read $3$ times monotone $2$\hy CNF} formulas.  
This nicely complements the known fact that read\hy once formulas have polynomial OBDD size~\cite{FFK88}; 
to the best of our knowledge, it was not even known that $3$ times formulas have super\hy polynomial OBDD size.
\end{itemize}


\medskip \noindent \textit{Organization.} The paper is organized as follows. In Section~\ref{sect:prel}, we introduce basic notation and terminology. 
In Section~\ref{sect:positive}, we prove that the few subterms property implies polynomial OBDD size for CNF classes, 
and prove that variable\hy convex CNFs (and bounded treewidth CNFs) have the few subterms property 
(fixed\hy parameter tractable size and time compilability results based on the few subterms property 
are presented in Section~\ref{sect:fpt}). In Section~\ref{sect:negative}, we prove an exponential lower bound 
on the OBDD size of CNF formulas based on expander graphs. Finally, we present our conclusions in Section~\ref{sect:conclusion}.

\section{Preliminaries}\label{sect:prel}
Let $X$ be a countable set of \emph{variables}. A \emph{literal} is a variable $x$ or a negated variable $\neg x$. If $x$ is a variable we let $\mathsf{var}(x) = \mathsf{var}(\neg x) = x$. A \emph{clause} is a finite set of literals. For a clause $c$ we define $\mathsf{var}(c) = \{ \var(l) \mid l \in c \}$. If a clause contains a literal negated as well as unnegated it is \emph{tautological}. A \emph{conjunctive normal form (CNF)} is a finite set of non\hy tautological clauses. If $F$ is a CNF formula we let $\mathsf{var}(F) = \bigcup_{c \in F} \mathsf{var}(c)$. The \emph{size} of a clause $c$ is $|c|$, and the \emph{size} of a CNF $F$ is $\mathsf{size}(F) = \sum_{c \in F} |c|$. An \emph{assignment} is a mapping $f \colon X' \to \{0,1\}$, where $X' \subseteq X$; we identify $f$ with the set $\{ \neg x \mid x \in X', f(x)=0 \} \cup \{ x \mid x \in X', f(x)=1 \}$. An assignment $f$ \emph{satisfies} a clause $c$ if $f \cap c \neq \emptyset$; for a CNF $F$, we let $F[f]$ denote the CNF containing the clauses in $F$ not satisfied by $f$, restricted to variables in $X \setminus \mathsf{var}(f)$, that is, $F[f] = \{ c \setminus \{ x, \neg x \mid x \in \mathsf{var}(f) \} \mid \text{$c \in F$, $f \cap c = \emptyset$} \}\text{;}$ then, $f$ \emph{satisfies} $F$ if $F[f] =\emptyset$, that is, if it satisfies all clauses in $F$. If $F$ is a CNF with $\var(F) = \{x_1, \dots, x_n\}$ we define the Boolean function $F(x_1, \dots, x_n)$ \emph{computed by $F$} as $F(b_1, \dots, b_n) = 1$ if and only if the assignment $f_{(b_1, \dots, b_n)}: \var(F) \rightarrow \{0, 1\}$ given by $f_{(b_1, \dots, b_n)}(x_i) = b_i$ satisfies the CNF $F$.

A \emph{binary decision diagram (BDD)} $D$ on variables $\{x_1, \ldots, x_n\}$ is a labelled directed acyclic graph satisfying the following conditions: $D$ has at at most two vertices without outgoing edges, called \emph{sinks} of $D$. Sinks of $D$ are labelled with $0$ or $1$; if there are exactly two sinks, one is labelled with $0$ and the other is labelled with $1$. Moreover, $D$ has exactly one vertex without incoming edges, called the \emph{source} of $D$. Each non\hy sink node of $D$ is labelled by a variable $x_i$, and has exactly two outgoing edges, one labelled $0$ and the other labelled $1$. Each node $v$ of $D$ represents a Boolean function $F_v = F_v(x_1, \ldots, x_n)$ in the following way. Let $(b_1, \dots, b_n) \in \{0, 1\}^n$ and let $w$ be a node labelled with $x_i$. We say that $(b_1, \dots, b_n)$ \emph{activates} an outgoing edge of $w$ labelled with $b \in \{0, 1\}$ if $b_i = b$. Since $(b_1, \dots, b_n)$ activates exactly one outgoing edge of each non\hy sink node, there is a unique sink that can be reached from $v$ along edges activated by $(b_1, \dots, b_n)$. We let $F_v(b_1, \dots, b_n) = b$, where $b \in \{0, 1\}$ is the label of this sink. The function \emph{computed by $D$} is $F_s$, where~$s$ denotes the (unique) source node of $D$. The \emph{size} of a BDD is the number of its nodes.

An \emph{ordering} $\sigma$ of a set $\{x_1,\ldots,x_n\}$ is a total order on $\{x_1,\ldots,x_n\}$. If $\sigma$ is an ordering of $\{x_1, \ldots, x_n\}$ we let $\mathsf{var}(\sigma) = \{x_1, \dots, x_n\}$. Let $\sigma$ be the ordering of $\{1,\ldots,n\}$ given by $x_{i_1}<x_{i_2}< \cdots <x_{i_n}$. For every integer $0 < j \leq n$, the \emph{length $j$ prefix} of $\sigma$ is the ordering of $\{x_{i_1},\ldots,x_{i_j}\}$ given by $x_{i_1}< \cdots <x_{i_j}$. A \emph{prefix} of $\sigma$ is a length $j$ prefix of $\sigma$ for some integer $0<j \leq n$.  For orderings $\sigma=x_{i_1}< \cdots <x_{i_n}$ of $\{x_1,\ldots,x_n\}$ and $\rho=y_{i_1}< \cdots <y_{i_m}$ of $\{y_1,\ldots,y_m\}$, we let $\sigma\rho$ denote the ordering of $\{x_1,\ldots,x_n,y_1,\ldots,y_m\}$ given by $x_{i_1}< \cdots <x_{i_n}<y_{i_1}< \cdots <y_{i_m}$.

Let $D$ be a BDD on variables $\{x_1, \ldots, x_n\}$ and let $\sigma = x_{i_1} < \cdots < x_{i_n}$ be an ordering of $\{x_1, \dots, x_n\}$. The BDD $D$ is a \emph{$\sigma$\hy ordered binary decision diagram ($\sigma$\hy OBDD)} if $x_i < x_j$ (with respect to $\sigma$) whenever $D$ contains an edge from a node labelled with $x_i$ to a node labelled with $x_j$. A BDD $D$ on variables $\{x_1, \ldots, x_n\}$ is an \emph{ordered binary decision diagram (OBDD)} if there is an ordering $\sigma$ of $\{x_1, \ldots,x_n\}$ such that $D$ is a $\sigma$\hy OBDD. For a Boolean function $F = F(x_1, \ldots, x_n)$, the \emph{OBDD size} of $F$ is the size of the smallest OBDD on $\{x_1, \ldots, x_n\}$ computing $F$.

We say that a class $\mathcal{F}$ of CNFs has \emph{polynomial\hy time compilation into OBDDs} if there is a polynomial\hy time algorithm that, given a CNF $F \in \mathcal{F}$, returns an OBDD computing the same Boolean function as $F$. We say that a class $\mathcal{F}$ of CNFs \emph{has polynomial size compilation into OBDDs} if there exists a polynomial $p \colon \mathbb{N} \to \mathbb{N}$ such that, for all CNFs $F \in \mathcal{F}$, there exists an OBDD of size at most $p(\mathsf{size}(F))$ that computes the same function as $F$.

For standard graph theoretic terminology, see~\cite{D2010}.  Let $G = (V, E)$ be a graph.  
The \emph{(open) neighborhood} of $W$ in $G$, in symbols $\mathsf{neigh}(W,G)$, is defined by 
\[\mathsf{neigh}(W,G) =\{ v \in V \setminus W \mid \text{there exists $w \in W$ such that $vw \in E$}\}\text{.}\]
We freely use $\mathsf{neigh}(v,G)$ as a shorthand for $\mathsf{neigh}(\{v\},G)$, 
and we write $\mathsf{neigh}(W)$ instead of $\mathsf{neigh}(W,G)$ if the graph $G$ is clear from the context. 

A graph $G = (V,E)$ is \emph{bipartite} if it its vertex set $V$ can be partitioned into two blocks $V'$ and $V''$ such that, for every edge $vw \in E$, we either have $v \in V'$ and $w \in V''$, or $v \in V''$ and  $w \in V'$. In this case we may write $G = (V', V'', E)$. The \emph{incidence} graph of a CNF $F$, in symbols $\mathsf{inc}(F)$, is the bipartite graph $(\mathsf{var}(F), F, E)$ such that $vc \in E$ if and only if $v \in \mathsf{var}(F)$, $c \in F$, and $v \in \mathsf{var}(c)$; that is, the blocks are the variables and clauses of $F$, and a variable is adjacent to a clause if and only if the variable occurs in the clause. 

A bipartite graph $G =(V,W,E)$ is \emph{left convex} if there exists an ordering $\sigma$ of $V$ such that the following holds: if $wv$ and $wv'$ are edges of $G$ and $v < v''  < v'$ (with respect to the ordering $\sigma$) then $wv''$ is an edge of $G$. The ordering~$\sigma$ is said to \emph{witness} left convexity of $G$. A CNF $F$ is \emph{variable convex} if $\mathsf{inc}(F) = (\mathsf{var}(F),F,E)$ is left convex. 

For an integer $d$, a CNF $F$ has \emph{degree} $d$ if $\mathsf{inc}(F)$ has degree at most $d$.  A class $\mathcal{F}$ of CNFs has \emph{bounded degree} if there exists an integer $d$ such that every CNF in $\mathcal{F}$ has degree $d$.  

\section{Polynomial Time Compilability}\label{sect:positive}

In this section, we introduce the \emph{few subterms} property, 
a sufficient condition for a class of CNFs to have polynomial size compilation into OBDDs (Section~\ref{sect:fewsubterms}).  
We prove that the classes of variable convex CNFs and bounded treewidth CNFs have the few subterms property 
(Section~\ref{sect:convex} and Section~\ref{sect:tw}).  Finally, we establish fixed-parameter tractable size 
and time OBDD compilation results for CNFs, where the parameter is the distance to a few subterms CNF class (Section~\ref{sect:fpt}).  

\subsection{Few Subterms}\label{sect:fewsubterms}

In this section, we introduce a property of classes of CNFs called the \emph{few subterms} property (Definition~\ref{def:fst}), 
and prove that classes of CNFs with the few subterms property admit polynomial time compilation into OBDDs (Corollary~\ref{cor:fewsubtermsptime}).

\begin{definition}[Few Subterms]\label{def:fst}
Let $F$ be a CNF, let $V \subseteq \mathsf{var}(F)$, 
and let $f \colon V \to \{0,1\}$.  The CNF $F[f]$ is called a \emph{$V$-subterm} of $F$.  
The set of $V$-subterms of $F$ is denoted by 
$$\mathsf{st}(F,V)=\{ F[f] \mid f \colon V \to \{0,1\} \}\text{.}$$
Let $\sigma$ be an ordering of $\mathsf{var}(F)$.  
The \emph{subterm width} of $F$ with respect to $\sigma$, in symbols $\mathsf{stw}(F,\sigma)$, is equal to 
$$\mathsf{stw}(F,\sigma)=\max \{ |\mathsf{st}(F,\mathsf{var}(\pi))| \mid \text{$\pi$ prefix of $\sigma$} \}\text{;}$$
the \emph{subterm width} of $F$ is the minimum subterm width of $F$ with respect to $\sigma$, 
where $\sigma$ ranges over all orderings of $\mathsf{var}(F)$.

Let $\mathcal{F}$ be a class of CNFs.  
A function $b \colon \mathbb{N} \to \mathbb{N}$ is called a \emph{subterm bound} of $\mathcal{F}$ 
if for all $F \in \mathcal{F}$, the subterm width of $F$ is bounded above by $b(\mathsf{size}(F))$.
Let $b \colon \mathbb{N} \to \mathbb{N}$ be a subterm bound of $\mathcal{F}$, 
let $F \in \mathcal{F}$, and let $\sigma$ be an ordering of $\mathsf{var}(F)$.  
We call $\sigma$ a \emph{witness} of subterm bound $b$ with respect to $F$ 
if $\mathsf{stw}(F,\sigma)\leq b(\mathsf{size}(F))$.
The class $\mathcal{F}$ has \emph{few subterms} if it has a polynomial subterm bound $p \colon \mathbb{N} \to \mathbb{N}$; 
if, in addition, for all $F \in \mathcal{F}$, 
an ordering $\sigma$ of $\mathsf{var}(F)$ witnessing $p$ with respect to $F$ 
can be computed in polynomial time, $\mathcal{F}$ is said to have \emph{constructive few subterms}.
\end{definition}

The following statement describes how the few subterms property 
naturally presents itself as a sufficient condition for a polynomial 
size construction of OBDDs from CNFs.  

\begin{theorem}\label{th:psizeimpliesptime}
  There exists 
  an algorithm that, given a CNF $F$ and an ordering $\sigma$ of
  $\mathsf{var}(F)$, returns a $\sigma$\hy OBDD for $F$ of size at most
  $|\mathsf{var}(F)|\: \mathsf{stw}(F, \sigma)$ in time polynomial in
  $|\mathsf{var}(F)|$ and $\mathsf{stw}(F, \sigma)$.
\end{theorem}
\begin{proof}[Proof of Theorem~\ref{th:psizeimpliesptime}]
  Let $F$ be a CNF and $\sigma=x_1\cdots x_n$ be an ordering of
  $\mathsf{var}(F)$.  The algorithm computes a $\sigma$-OBDD $D$ for $F$ as
  follows.

  At step $i=1$, create the source of $D$, labelled by $F$, at the level $0$
  of the diagram; if $\emptyset \in F$ (respectively, $F = \emptyset$), then
  identify the source with the $0$-sink (respectively, $1$-sink) of the
  diagram, otherwise make the source an $x_1$-node.

  At step $i+1$ for $i=1,\ldots,n-1$, let $v_1,\ldots,v_l$ be the
  $x_{i}$-nodes at level $i-1$ of the diagram, respectively labelled
  $F_1,\ldots,F_l$.  For $j=1,\ldots,l$ and $b=0,1$, compute $F_j[x_{i}=b]$,
  where $x_i=b$ denotes the assignment $f \colon \{x_i\} \to \{0,1\}$ mapping
  $x_i$ to $b$.  If $F_j[x_{i}=b]$ is equal to some label of an $x_{i+1}$-node
  $v$ already created at level $i$, then direct the $b$-edge leaving the
  $x_{i}$-node labelled $F_j$ to $v$; otherwise, create a new $x_{i+1}$-node
  $v$ at level $i$, labelled $F_j[x_{i}=b]$, and direct the $b$-edge leaving
  the $x_{i}$-node labelled $F_j$ to $v$.  If $\emptyset \in F_j[x_{i}=b]$,
  then identify $v$ with the $0$-sink of $D$, and if $\emptyset=F_j[x_{i}=b]$,
  then identify $v$ with the $1$-sink of $D$.

  At termination, the diagram obtained computes $F$ and respects $\sigma$.  We
  analyze the runtime.  At step $i+1$ ($0 \leq i<n$), the nodes created at
  level $i$ are labelled by CNFs of the form $F[f]$, where $f$ ranges over all
  assignments of $\{x_1,\ldots,x_i\}$ not falsifying $F$; that is, these nodes
  correspond exactly to the $\{x_1,\ldots,x_i\}$-subterms
  $\mathsf{st}(F,\{x_1,\ldots,x_i\})$ of $F$ not containing the empty clause,
  whose number is bounded above by $\mathsf{stw}(F, \sigma)$. As level $i$ is
  processed in time bounded above by its size times the size of level $i-1$,
  and $|\mathsf{var}(F)|$ levels are processed, the diagram $D$ 
has size at most
  $|\mathsf{var}(F)| \cdot \mathsf{stw}(F, \sigma)$ and is constructed
  in time bounded above by a polynomial in $|\mathsf{var}(F)|$ and
  $\mathsf{stw}(F, \sigma)$.
\end{proof}

\begin{corollary}\label{cor:fewsubtermsptime}
Let $\mathcal{F}$ be a class of CNFs with constructive few subterms.  
Then $\mathcal{F}$ has has polynomial time compilation into OBDDs.
\end{corollary}
\begin{proof}[Proof of Corollary~\ref{cor:fewsubtermsptime}]
  Let $\mathcal{F}$ be a class of CNFs with constructive few subterms, and let
  $p \colon \mathbb{N} \to \mathbb{N}$ be a polynomial subterm bound of
  $\mathcal{F}$. The algorithm, given a CNF $F$, computes in polynomial time
  an ordering of $\mathsf{var}(F)$ witnessing $p$ with respect to $F$, and
  invokes the algorithm in Theorem~\ref{th:psizeimpliesptime}, which runs in
  time polynomial in $|\mathsf{var}(F)|$ and $\mathsf{stw}(F, \sigma)$. Since
  $\mathsf{stw}(F, \sigma) \leq p(\mathsf{size}(F))$ the overall runtime is polynomial in
  $\mathsf{size}(F)$.
\end{proof}

\subsection{Variable Convex}\label{sect:convex}

In this section, we prove that the class of variable convex CNFs has the constructive few subterms property
(Lemma~\ref{lemma:convexsmallproj}), and hence admits polynomial time
compilation into OBDDs (Theorem~\ref{th:leftconvex}); as a special case, CNFs
whose incidence graphs are cographs admit polynomial time compilation into
OBDDs (Example~\ref{ex:cographs}).

\begin{lemma}\label{lemma:convexsmallproj}
  The class $\mathcal{F}$ of variable convex CNFs has the constructive few
  subterms property.
\begin{proof}
  Let $F \in \mathcal{F}$, so that $\mathsf{inc}(F)$ is left convex, and let 
  $\sigma$ be an ordering of $\mathsf{var}(F)$ witnessing the left convexity 
  of $\mathsf{inc}(F)$.
Let $\pi$ be any prefix of $\sigma$.  Call a clause $c \in F$ \emph{$\mathsf{var}(\pi)$-active} in $F$ if 
$\mathsf{var}(c) \cap \mathsf{var}(\pi) \neq \emptyset$ and $\mathsf{var}(c) \cap (\mathsf{var}(F)\setminus \mathsf{var}(\pi)) \neq \emptyset$.  
Let $\mathsf{ac}(F,\mathsf{var}(\pi))$ denote the CNF containing the $\mathsf{var}(\pi)$-active clauses of $F$.  
For all $c \in \mathsf{ac}(F,\mathsf{var}(\pi))$, let $\mathsf{var}(c)'=\mathsf{var}(c) \cap \mathsf{var}(\pi)$.  

\begin{claim}
Let $c,c' \in \mathsf{ac}(F,\mathsf{var}(\pi))$.  Then, $\mathsf{var}(c)' \subseteq \mathsf{var}(c')'$ 
or $\mathsf{var}(c')' \subseteq \mathsf{var}(c)'$.
\begin{proof}[Proof of Claim]
Let $c,c' \in \mathsf{ac}(F,\mathsf{var}(\pi))$.  Assume for a contradiction that the statement does not hold, 
that is, there exist variables $v,v' \in \mathsf{var}(\pi)$, $v \neq v'$, 
such that $v \in \mathsf{var}(c)'\setminus \mathsf{var}(c')'$ and $v' \in \mathsf{var}(c')'\setminus \mathsf{var}(c)'$.  
Assume that $\sigma(v)<\sigma(v')$; the other case is symmetric.  Since $c \in \mathsf{ac}(F,\mathsf{var}(\pi))$, 
by definition there exists a variable $w \in \mathsf{var}(F)\setminus \mathsf{var}(\pi)$ such that $w \in \mathsf{var}(c)$.  
It follows that $\sigma(v')<\sigma(w)$.  Therefore, 
we have $\sigma(v)<\sigma(v')<\sigma(w)$, 
where $v,w \in \mathsf{var}(c)$ and $v' \not\in \mathsf{var}(c)$, contradicting the fact that $\sigma$ witnesses the left convexity of $\mathsf{inc}(F)$.
\end{proof}
\end{claim}

We now introduce a partially ordered set $P$, representing the entailment relation 
among $\mathsf{var}(\pi)$-active clauses restricted to literals on variables in $\mathsf{var}(\pi)$.  Formally, 
we define $P$ as follows:
\begin{itemize}
\item The elements are equivalence classes $[c]_\equiv$ of the equivalence relation 
on $\mathsf{ac}(F,\mathsf{var}(\pi))$ defined as follows.  For all $c,c' \in \mathsf{ac}(F,\mathsf{var}(\pi))$, $c \equiv c'$ if and only if, for all $v \in \mathsf{var}(\pi)$,  
$$\{v,\neg v\} \cap c=\{v,\neg v\} \cap c'\text{;}$$ 
in words, literals on $v$ occur identically in $c$ and $c'$.
\item The order is defined by putting, for all elements $[c]_\equiv$ and $[c']_\equiv$, 
$$\text{$[c]_\equiv \leq [c']_\equiv$ if and only if $\{ l \in c \mid \mathsf{var}(l) \in \mathsf{var}(\pi)\} \subseteq \{ l \in c' \mid \mathsf{var}(l) \in \mathsf{var}(\pi)\}$;}$$
in words, upon restriction to literals on variables in $\mathsf{var}(\pi)$, every clause in $[c]_\equiv$ entails every clause in $[c']_\equiv$.
\end{itemize}
Observe that $|P| \leq |F|$, because $|P| \leq |\mathsf{ac}(F,\mathsf{var}(\pi))|$ and $\mathsf{ac}(F,\mathsf{var}(\pi)) \subseteq F$.

We now establish a correspondence between the $\mathsf{var}(\pi)$-subterms of $F$ 
and the elements in $P$, which allows to bound above the size of $\mathsf{st}(F,\mathsf{var}(\pi))$ by the size of $P$.  

\begin{claim}
Let $f \colon \mathsf{var}(\pi) \to \{0,1\}$ be an assignment not satisfying $\mathsf{ac}(F,\mathsf{var}(\pi))$. 
\begin{itemize}
\item There exists $c \in \mathsf{ac}(F,\mathsf{var}(\pi))$ such that $[c]_\equiv$ is maximal in $P$ 
with the property that $f$ does not satisfy $c$. 
\item Let $t \in \mathsf{ac}(F,\mathsf{var}(\pi))$ be such that $[t]_\equiv$ is maximal in $P$ 
with the property that $f$ does not satisfy $t$.  
Then, 
$$\mathsf{ac}(F,\mathsf{var}(\pi))[f]=\{ c \in [s]_\equiv \mid [s]_\equiv \leq [t]_\equiv \}[f]\text{.}$$
\end{itemize}
\begin{proof}[Proof of Claim]
For the first part, let $f \colon \mathsf{var}(\pi) \to \{0,1\}$ be an assignment not satisfying $\mathsf{ac}(F,\mathsf{var}(\pi))$. 
By the first claim, there is a unique inclusion maximal clause $c$ among the clauses in $\mathsf{ac}(F,\mathsf{var}(\pi))$ not satisfied by $f$.  
If $[c]_\equiv$ is maximal in $P$, then we are done.  Otherwise, 
assume that $[c]_\equiv$ is not maximal in $P$, 
and assume for a contradiction that there exists $[d]_\equiv$ 
such that $[c]_\equiv < [d]_\equiv$ and $f$ does not satisfy $d$.  
Since $[c]_\equiv < [d]_\equiv$, it holds that $d$ contains 
at least one literal $l$, on a variable in $\mathsf{var}(\pi)$, such that $l$ is not in $c$; 
a contradiction, since $c$ is chosen inclusion maximal among the clauses in $\mathsf{ac}(F,\mathsf{var}(\pi))$ not satisfied by $f$.

For the second part, let $t \in \mathsf{ac}(F,\mathsf{var}(\pi))$ be such $[t]_\equiv$ is maximal in $P$ 
with the property that $f$ does not satisfy $t$.  By definition, 
if $c \in [s]_\equiv$ and $[s]_\equiv \leq [t]_\equiv$, 
then $c$ entails $t$ upon restriction to variables in $\mathsf{var}(\pi)$.  Hence, 
if $f$ does not satisfy $t$, it holds that $f$ does not satisfy $c$.  Hence, 
$\mathsf{ac}(F,\mathsf{var}(\pi))[f]$ is equal to the clauses in $\{ c \in [s]_\equiv \mid [s]_\equiv \leq [t]_\equiv \}$, 
restricted to variables not in $\mathsf{var}(\pi)$, which is exactly the effect of applying $f$ to 
$\{ c \in [s]_\equiv \mid [s]_\equiv \leq [t]_\equiv \}$.

The claim is settled.
\end{proof}
\end{claim}

Let $f \colon \mathsf{var}(\pi) \to \{0,1\}$ be any assignment.  
If $f$ does not satisfy $\mathsf{ac}(F,\mathsf{var}(\pi))$, 
then, by the second claim, $\mathsf{ac}(F,\mathsf{var}(\pi))[f]$ corresponds to an element in $P$.  
Hence, the number of $\mathsf{var}(\pi)$-subterms of $F$ generated by assignments not satisfying $\mathsf{ac}(F,\mathsf{var}(\pi))$ 
is bounded above by the number of elements in $P$; we observed that $|P| \leq |F|$.  Otherwise, if $f$ satisfies $\mathsf{ac}(F,\mathsf{var}(\pi))$, 
then $\mathsf{ac}(F,\mathsf{var}(\pi))[f]=\emptyset$.  It follows that 
$$|\mathsf{st}(\mathsf{ac}(F,\mathsf{var}(\pi)))| \leq |F|+1 \leq \mathsf{size}(F)+1\text{.}$$

Note that $F$ is the disjoint union of $\mathsf{ac}(F,\mathsf{var}(\pi))$, 
clauses $C' \subseteq F$ whose variables are all in $\mathsf{var}(\pi)$, 
and clauses $C'' \subseteq F$ whose variables are all outside $\mathsf{var}(\pi)$.  
Also, $\mathsf{st}(C',\mathsf{var}(\pi)) \subseteq \{\emptyset,\{\emptyset\}\}$, 
and $\mathsf{st}(C'',\mathsf{var}(\pi))=\{C''\}$.  It follows that 
$$|\mathsf{st}(F,\mathsf{var}(\pi))| \leq |\mathsf{st}(\mathsf{ac}(F,\mathsf{var}(\pi)))| \cdot |\mathsf{st}(C',\mathsf{var}(\pi))| \cdot |\mathsf{st}(C'',\mathsf{var}(\pi))| \leq 2(\mathsf{size}(F)+1)\text{.}$$

This shows that $\mathsf{stw}(F, \sigma)$ is linear in the size of $F$, where
$\sigma$ is an ordering witnessing left convexity of $\mathsf{inc}(F)$. This
proves that the class of variable convex CNFs has few subterms. Moreover, an
ordering witnessing the left convexity of $\mathsf{inc}(F)$ can be computed in
polynomial (even linear) time~\cite{BL76,KKLV11}, so the class of variable convex CNFs even has
the constructive few subterms property.
\end{proof}
\end{lemma}

\begin{theorem}\label{th:leftconvex}
  The class of variable convex CNF formulas has polynomial time compilation
  into OBDDs.
\begin{proof}
  Immediate from Corollary~\ref{cor:fewsubtermsptime} and
  Lemma~\ref{lemma:convexsmallproj}.\end{proof}
\end{theorem}

\begin{example}[Bipartite Cographs]\label{ex:cographs}
Let $F$ be a CNF such that $\mathsf{inc}(F)$ is a cograph.  Note 
that $\mathsf{inc}(F)$ is a complete bipartite graph.  Indeed, cographs 
are characterized as graphs of clique-width at most $2$ \cite{CO00}, 
and it is readily verified that if a bipartite graph has clique\hy width at most $2$, 
then it is a complete bipartite graph.  
A complete bipartite graph is trivially left convex.  
Then Theorem~\ref{th:leftconvex} implies that CNFs whose incidence graphs are cographs 
have polynomial time compilation into OBDDs.  
\end{example}

\subsection{Bounded Treewidth}\label{sect:tw}

In this section, we prove that if a class of CNFs has \emph{bounded
  treewidth}, then it has the constructive few subterms property
(Lemma~\ref{lemma:twfewsubterms}), and hence admits polynomial time
compilation into OBDDs (Theorem~\ref{th:twpolytime}).  

\medskip \noindent A \emph{tree decomposition} $T$ of a graph $G$ 
is a rooted tree whose elements, called \emph{bags}, 
are subsets of the vertices of $G$ satisfying the following:
\begin{itemize}
\item for every vertex $v$ of $G$, there is a bag containing $v$;
\item for every edge $vw$ of $G$, there is a bag containing $v$ and $w$;
\item for every three bags $B, B',B'' \in P$, if $B \leq B' \leq B''$, 
then $B \cap B'' \subseteq B'$.
\end{itemize}
A graph $G$ has \emph{treewidth} $k$ if it has a tree decomposition $T$ 
such that each bag contains at most $k+1$ vertices; $T$ is said to \emph{witness} treewidth $k$ for $G$.  
The notions of \emph{path decomposition} and \emph{pathwidth} are defined analogously 
using paths instead of trees.  

Let $F$ be a CNF.  We say that $\mathsf{inc}(F)=(\mathsf{var}(F),F,E)$ has treewidth (respectively, pathwidth) $k$ 
if the graph $(\mathsf{var}(F) \cup F,E)$ has treewidth (respectively, pathwidth) $k$.  
We identify the pathwidth (respectively, treewidth) of a CNF 
with the pathwidth (respectively, treewidth) of its incidence graph.  If $\mathsf{inc}(F)$ has pathwidth $k$, 
then an ordering $\sigma$ of $\mathsf{var}(F)$ is called a \emph{forget} ordering for $F$ if, 
with respect to an arbitrary linearization of some path decomposition witnessing pathwidth $k$ for $\mathsf{inc}(F)$, 
if the first bag containing $v$ is less than or equal to the first bag containing $v'$, 
then $\sigma(v)<\sigma(v')$. 

A proof of the following lemma already appears, in essence, in previous work by 
Ferrara, Pan, and Vardi \cite[Theorem~2.1]{FPV05} and Razgon \cite[Lemma~5]{R14}.

\begin{lemma}\label{lemma:activebounded}
  Let $F$ be a CNF of 
  pathwidth $k-1$, and let $\sigma$ be a forget ordering for~$F$.  Then
  $\mathsf{stw}(F, \sigma) \leq 2^{k+1}$.
\begin{proof}
  Let $F$ be a CNF such that $\mathsf{inc}(F)$ has pathwidth $k-1$, let
  $\sigma$ be a forget ordering for~$F$, and let $\pi$ be any prefix of
  $\sigma$.

  Let $v$ be the last variable in $\mathsf{var}(\pi)$ relative to the ordering
  $\sigma$, and let $B$ be the first bag (in the total order of $P$) that
  contains $v$.  A clause $c \in F$ is called
  \emph{$\mathsf{var}(\pi)$-active} in $F$ if $\mathsf{var}(c) \cap
  \mathsf{var}(\pi) \neq \emptyset$ and $\mathsf{var}(c) \cap
  (\mathsf{var}(F)\setminus \mathsf{var}(\pi)) \neq
  \emptyset$.  
  Let $\mathsf{ac}(F,\mathsf{var}(\pi))$ denote the CNF containing the
  $\mathsf{var}(\pi)$-active clauses of $F$.  Let
\begin{align*}
C' &= \mathsf{ac}(F,\mathsf{var}(\pi)) \cap B\text{,}\\
C'' &= \{ c \in \mathsf{ac}(F,\mathsf{var}(\pi)) \mid \text{$c \in B'$ only if $B'>B$ in $P$} \}\text{;}
\end{align*}
in words, $C'$ contains $\mathsf{var}(\pi)$-active clauses in the bag $B$, 
and $C''$ contains $\mathsf{var}(\pi)$-active clauses occurring only in bags strictly larger than $B$ in the total order of $P$.  
Clearly, $C' \cap C''=\emptyset$.

\begin{claim}
$\mathsf{ac}(F,\mathsf{var}(\pi))=C' \cup C''$. 
\begin{proof}[Proof of Claim]
First observe that a $\mathsf{var}(\pi)$-active clause $c$ cannot occur only in bags strictly smaller than $B$ in the total order of $P$.  
For otherwise, since $\mathsf{var}(c) \cap (\mathsf{var}(F) \setminus \mathsf{var}(\pi)) \neq \emptyset$, let $v' \in \mathsf{var}(c) \cap (\mathsf{var}(F) \setminus \mathsf{var}(\pi))$; 
if $B'$ is the first bag that contains $v'$, 
then $B \leq B'$ (by the choice of $v$), hence $v'$ is not contained in any bag strictly smaller than $B$, 
and the edge $cv'$ is not witnessed in $P$, a contradiction.  

Thus $\mathsf{var}(\pi)$-active clauses either occur in $B$ (including the case where 
they occur in $B$ and in bags smaller or larger than $B$ in $P$), 
or occur only in bags strictly larger than $B$ in $P$.  Thus, $\mathsf{ac}(F,\mathsf{var}(\pi)) \subseteq C' \cup C''$; 
the other inclusion holds by definition.
\end{proof}
\end{claim}

The claim and the fact that $C' \cap C''=\emptyset$ imply that 
$$|\mathsf{st}(\mathsf{ac}(F,\mathsf{var}(\pi)))| \leq |\mathsf{st}(C',\mathsf{var}(\pi))| \cdot |\mathsf{st}(C'',\mathsf{var}(\pi))|\text{;}$$
thus, suffices to bound above the size of the two sets on the right 
so that the product of the individual bounds is at most $2^k$.  Let $k'=|C'|$.  Obviously,

\begin{claim}
$|\mathsf{st}(C',\mathsf{var}(\pi))|\leq 2^{k'}$.
\end{claim}

Let $V'=\bigcup_{c \in C''} \mathsf{var}(c) \cap \mathsf{var}(\pi)$ and let $k''=|V'|$.

\begin{claim}
$V' \subseteq B$. 
\begin{proof}[Proof of Claim]
Let $c$ be a $\mathsf{var}(\pi)$-active clause occurring only in bags strictly larger than $B$ in $P$.  
Let $v' \in \mathsf{var}(c) \cap \mathsf{var}(\pi)$.  By the choice of $v$ and the properties of the forget ordering $\sigma$, 
it holds that the first bag containing $v'$ is less than or equal to $B$.  
Since $B$ is the first bag that contains $v$, it holds that $v' \in B$ by the properties of $P$ 
(the edge $cv'$ is witnessed in a bag strictly larger than $B$ in $P$).
\end{proof}
\end{claim}

\begin{claim}
$|\mathsf{st}(C'',\mathsf{var}(\pi))|\leq 2^{k''}$.
\begin{proof}[Proof of Claim]
Define an equivalence relation on $\mathsf{var}(\pi)$-assignments as follows: 
For all $f,f' \colon \mathsf{var}(\pi) \to \{0,1\}$, 
$f \equiv f'$ if and only if, for all $v \in V'$, 
$f(v)=f'(v)$.  Since $|V'|=k''$, the equivalence relation 
has $2^{k''}$ many equivalence classes.  Moreover, 
if $f \equiv f'$, then $C''[f]=C''[f']$, 
because $\mathsf{var}(C'') \subseteq V'$.  The claim follows.
\end{proof}
\end{claim}

Since $C',V' \subseteq B$ and $C' \cap V'=\emptyset$, it holds that $k'+k''=|C'|+|V'| \leq |B| \leq k$.  
Hence, $$|\mathsf{st}(\mathsf{ac}(F,\mathsf{var}(\pi)))| \leq 2^{k'} \cdot 2^{k''}=2^{k'+k''} \leq 2^k\text{.}$$

Note that $F$ is the disjoint union of $\mathsf{ac}(F,\mathsf{var}(\pi))$, 
clauses $D' \subseteq F$ whose variables are all in $\mathsf{var}(\pi)$, 
and clauses $D'' \subseteq F$ whose variables are all outside $\mathsf{var}(\pi)$.  
Also, $\mathsf{st}(D',\mathsf{var}(\pi)) \subseteq \{\emptyset,\{\emptyset\}\}$, 
and $\mathsf{st}(D'',\mathsf{var}(\pi))=\{C''\}$.  It follows that 
$$|\mathsf{st}(F,\mathsf{var}(\pi))| \leq |\mathsf{st}(\mathsf{ac}(F,\mathsf{var}(\pi)))| \cdot |\mathsf{st}(D',\mathsf{var}(\pi))| \cdot |\mathsf{st}(D'',\mathsf{var}(\pi))| \leq 2^{k+1}\text{.}$$

and the statement is proved.  
\end{proof}
\end{lemma}

\begin{lemma}\label{lemma:twfewsubterms}
  Let $\mathcal{F}$ be a class of CNFs of bounded treewidth. Then
  $\mathcal{F}$ has the constructive few subterms property.
\begin{proof}
  Let $c-1$ be the treewidth bound of $\mathcal{F}$ and let $F \in
  \mathcal{F}$, so that the treewidth of $\mathsf{inc}(F)$ is at most
  $c-1$. We can compute a width $c-1$ tree decomposition of $\mathsf{inc}(F)$
  in linear time $O(\mathsf{size}(F))$~\cite{BodlaenderKloks96}. From this decomposition, we can
  compute a path decomposition of $\mathsf{inc}(F)$ of width at most $(c-1)
  \cdot \log |\mathsf{var}(F) \cup F| \leq c \cdot \log |\mathsf{var}(F) \cup
  F|-1$~\cite[Corollary~24]{Bodlaender1996Arboretum} and a corresponding
  forget ordering of $\mathsf{var}(F)$ in polynomial time. By
  Lemma~\ref{lemma:activebounded}, the subterm width of $F$ with respect to
  $\sigma$ is at most $2^{c \cdot \log |\mathsf{var}(F) \cup
    F|}=|\mathsf{var}(F) \cup F|^c \leq O(\mathsf{size}(F)^c)$. Thus $\mathcal{F}$ has a
  polynomial subterm bound, and a witnessing ordering $\sigma$ can be computed
  for each $F \in \mathcal{F}$ in polynomial time. We conclude that
  $\mathcal{F}$ has the constructive few subterms property.
\end{proof}
\end{lemma}

\begin{theorem}\label{th:twpolytime}
Let $\mathcal{F}$ be a class of CNFs of bounded treewidth.  Then, 
$\mathcal{F}$ has polynomial time compilation into OBDDs.
\begin{proof}
  Immediate from Lemma~\ref{lemma:twfewsubterms} and
  Corollary~\ref{cor:fewsubtermsptime}.
\end{proof}
\end{theorem}

\subsection{Almost Few Subterms}\label{sect:fpt}

In this section, we use the language of \emph{parameterized complexity} 
to formalize the observation that CNF classes \lq\lq close\rq\rq\ to CNF classes 
with few subterms have \lq\lq small\rq\rq\ OBDD representations~\cite{DowneyFellows13,FlumGrohe}.  

\medskip \noindent   Let $F$ be a CNF and $D$ a set of variables and clauses of~$F$. Let $E$ be 
the formula obtained by deleting $D$ from~$F$, that is, 
$$E=\{ c \setminus \{ l \in c \mid \mathsf{var}(l) \in D \} \mid c \in F \setminus D \}\text{;}$$
we call $D$ the \emph{deletion set} of $F$ with respect to $E$.

The following lemma shows that adding a few variables and clauses 
does not increase the subterm width of a formula too much.

\begin{lemma}\label{lemma:deletion}
  Let $F$ and $E$ be CNFs such that $D$ is the deletion set of~$F$ with
  respect to~$E$. Let~$\pi$ be an ordering of $\mathsf{var}(E)$ and let
  $\sigma$ be an ordering of $\mathsf{var}(F) \cap D$. Then $\mathsf{stw}(F,
  \sigma \pi) \leq 2^k \cdot \mathsf{stw}(E, \pi)$.
\end{lemma}
\begin{proof}
  Let $V = D \cap \mathsf{var}(F)$ and $C = D \cap F$, and let $k' =
  |V|$ and $k'' = |C|$. Let $\rho$ be a prefix of $\sigma \pi$ and $X
  = \mathsf{var}(\rho)$. From $F = C \cup (F \setminus C)$ we get
  $\mathsf{st}(F, X) = \mathsf{st}(C, X) \cup \mathsf{st}(F \setminus
  C, X)$, which allows us to bound the number of subterms
  $|\mathsf{st}(F, X)|$ as
  \begin{align}
    |\mathsf{st}(F, X)| \leq |\mathsf{st}(C, X)|\cdot |\mathsf{st}(F
    \setminus C, X)|. \label{eq:prodbound}
  \end{align}
  The number of subterms $C[f]$ for $f \in \{0, 1\}^X$ is bounded from
  above by the number of subsets of $C$, so $\mathsf{st}(C, X) \leq
  2^{k''}$. Recall that $\mathsf{st}(F \setminus C, X) = \SB (F
  \setminus C)[f]\:|\:f \in \{0, 1\}^X \SE$. Splitting the assignments
  $f$ into two parts, we can write this as
  \begin{align} 
    \mathsf{st}(F \setminus C, X) = \SB (F \setminus
    C)[f'][f'']\:|\:f' \in \{0, 1\}^{V \cap X}, f'' \in \{0, 1\}^{X
      \setminus V} \SE. \label{eq:stbound}
\end{align}
Let $f' \in \{0, 1\}^{V \cap X}$ be an assignment. The formula $E$ is
obtained from $F \setminus C$ by deleting variables in $V$. It follows
that $(F \setminus C)[f'] \subseteq E$ and so $(F \setminus
C)[f'][f''] \subseteq E[f'']$ for any assignment $f'' \in \{0, 1\}^{X
  \setminus V}$. This yields
\begin{align}
  |\SB (F \setminus C)[f'][f'']\:|\:f'' \in \{0, 1\}^{X \setminus V}
  \SE| \leq |\SB E[f'']\:|\:f'' \in \{0, 1\}^{X \setminus V}
  \SE|, \label{eq:acbound}
\end{align}
and the right hand side of this inequality corresponds to
$|\mathsf{st}(E, X \setminus V)|$. Combining this
with~(\ref{eq:stbound}), we obtain
\begin{align*} 
  |\mathsf{st}(F \setminus C, X)| &= |\SB (F \setminus C)[f'][f'']\:|\:f'
  \in \{0, 1\}^{V \cap X}, f'' \in \{0, 1\}^{X \setminus V} \SE| \\
  &\leq 2^{k'} |\SB E[f'']\:|\:f'' \in \{0, 1\}^{X \setminus V} \SE| = 2^{k'} \cdot  |\mathsf{st}(E, X \setminus V)| \\
  &\leq 2^{k'}\cdot  \mathsf{stw}(E, \pi).
\end{align*}
Inserting into~(\ref{eq:prodbound}), we get 
\begin{align*} |\mathsf{st}(F, X)| \leq
|\mathsf{st}(C, X)|\cdot |\mathsf{st}(F \setminus C, X)| \leq 2^{k''} \cdot  2^{k'}\cdot \mathsf{stw}(E, \pi) = 2^k \cdot  \mathsf{stw}(E, \pi),
\end{align*}
where $k=k'+k''$, and the lemma is proved.
\end{proof}

In this section, the standard of efficiency we appeal to comes from the framework of \emph{parameterized complexity}~\cite{DowneyFellows13,FlumGrohe}.  
The parameter we consider is defined as follows.  Let $\mathcal{F}$ be a class of CNF formulas. We say that 
$\mathcal{F}$ is \emph{closed under variable and clause deletion} if 
$E \in \mathcal{F}$ whenever $E$ is obtained by deleting variables or
clauses from $F \in \mathcal{F}$. Let $\mathcal{F}$ be a CNF class closed under variable and clause deletion. 
The \emph{$\mathcal{F}$\hy deletion distance of $F$} is the minimum size of a deletion set of $F$ 
from any $E \in \mathcal{F}$.  An $\mathcal{F}$\hy deletion set of $F$ is a deletion set of $F$ with respect to some $E \in \mathcal{F}$.

Let $\mathcal{F}$ be a class of CNF formulas with few subterms closed under variable and clause deletion.  
We say that CNFs have fixed-parameter tractable OBDD size, parameterized by $\mathcal{F}$\hy deletion distance, 
if there is a computable function $f:\mathbb{N} \rightarrow \mathbb{N}$, 
a polynomial $p: \mathbb{N} \rightarrow \mathbb{N}$, and an algorithm that, 
given a CNF $F$ having $\mathcal{F}$\hy deletion distance $k$, 
computes an OBDD equivalent to $F$ in time $f(k)\:p(\mathsf{size}(F))$.  

\begin{theorem}\label{th:fptsizecompiln}
Let $\mathcal{F}$ be a class of CNF formulas with few subterms closed under variable and clause deletion.  
CNFs have fixed-parameter tractable OBDD size parameterized by $\mathcal{F}$\hy deletion distance.
\end{theorem}

The assumption that $\mathcal{F}$ is closed under variable and clause deletion is technically necessary to have, 
for every CNF, a finite deletion distance from $\mathcal{F}$; it is a mild assumption though, 
as it is readily verified that if $\mathcal{F}$ has few subterms with polynomial subterm bound $p: \mathbb{N} \rightarrow \mathbb{N}$, 
then also the closure of $\mathcal{F}$ under variable and clause deletion has few subterms with the same polynomial subterm bound.  

\begin{proof}
Let $\mathcal{F}$ be a class of CNF formulas with few subterms closed under variable and clause deletion.  
Since $\mathcal{F}$ has few subterms, it has a polynomial subterm 
  bound $p: \mathbb{N} \rightarrow \mathbb{N}$. Let $k$ be the 
  $\mathcal{F}$\hy deletion distance of $F$. Let $E \in \mathcal{F}$ be a
  formula such that the deletion distance of~$F$ from $E$ is $k$, and let $D$
  the deletion set of $F$ with respect to $E$. Let $\pi$ be an ordering of
  $\mathsf{var}(E)$ witnessing $p$ for $E$, and let $\sigma$ be an ordering of
  $\mathsf{var}(F) \cap D$. By Lemma~\ref{lemma:deletion}, the subterm width
  of $F$ with respect to $\rho = \sigma \pi$ is at most $2^k p(\mathsf{size}(E))$, so by
  Theorem~\ref{th:psizeimpliesptime} there is a $\rho$\hy OBDD for $F$ of size
  at most $2^k p(\mathsf{size}(E))\:|\mathsf{var}(F)|$.  
\end{proof}

Analogously, we say that CNFs have fixed-parameter tractable time computable OBDDs (respectively, $\mathcal{F}$\hy deletion sets), 
parameterized by $\mathcal{F}$\hy deletion distance, if an OBDD (respectively, a $\mathcal{F}$\hy deletion set) for a 
given CNF~$F$ of $\mathcal{F}$\hy deletion distance $k$ is computable in time bounded above by $f(k)\:p(\mathsf{size}(F))$.  

\begin{theorem}\label{th:fptcompil}
  Let $\mathcal{F}$ be a class of CNFs closed under variable and clause deletion satisfying the following:
  \begin{itemize} \item $\mathcal{F}$ has the constructive few
    subterms property.
  \item CNFs have fixed-parameter tractable time computable $\mathcal{F}$\hy deletion sets, 
parameterized by $\mathcal{F}$\hy deletion distance.
\end{itemize}
CNFs have fixed-parameter tractable time computable OBDDs 
parameterized by $\mathcal{F}$\hy deletion distance.
\begin{proof}
  Given an input formula $F$, the algorithm first computes a smallest
  $\mathcal{F}$\hy deletion set~$D$ of $F$. Let $E$ be the formula obtained
  from $F$ by deleting $D$. The algorithm computes a variable ordering $\pi$
  of $E$ witnessing a polynomial subterm bound $p: \mathbb{N} \rightarrow
  \mathbb{N}$ of $\mathcal{F}$. Since~$\mathcal{F}$ has the constructive few
  subterms property, this can be done in polynomial time. Next, the algorithm
  chooses an arbitrary ordering $\sigma$ of $\mathsf{var}(F) \cap D$. By
  Lemma~\ref{lemma:deletion} we have $\mathsf{stw}(F, \sigma \pi) \leq
  2^{|D|}\:\mathsf{stw}(E, \pi) \leq 2^k\:p(\mathsf{size}(E))$, where $k$ is the
  $\mathcal{F}$\hy deletion distance of $F$. Invoking the algorithm of
  Theorem~\ref{th:psizeimpliesptime}, our algorithm computes and returns an
  OBDD for $F$ in time polynomial in $2^k\:p(\mathsf{size}(E))\:|\mathsf{var}(F)|$. Since
  $\mathsf{size}(E) \leq \mathsf{size}(F)$ there is a polynomial $q: \mathbb{N} \rightarrow
  \mathbb{N}$ such that last expression is bounded by $2^k q(\mathsf{size}(F))$.
\end{proof}
\end{theorem}

\begin{corollary}[Feedback Vertex Set]\label{cor:fvs}
  Let $\mathcal{F}$ be the class of formulas whose incidence graphs are
  forests. CNFs have fixed-parameter tractable time computable OBDDs 
parameterized by $\mathcal{F}$\hy deletion distance.
\end{corollary}
\begin{proof}
  Given a graph $G=(V,E)$, a set $D \subseteq V$ is called a \emph{feedback
    vertex set} of $G$ if the graph $G \setminus D$ is a 
forest; here, $G \setminus D$ is the graph $(V \setminus D, E')$ such that $vw \in E'$ if and only if $vw \in E$ and $v,w \in V \setminus D$. 
For any CNF $F$, a subset $D$ of its variables and clauses is a
  feedback vertex set of the incidence graph $\mathsf{inc}(F)$ if and only if
  it is a $\mathcal{F}$\hy deletion set, so a smallest feedback vertex set of
  $\mathsf{inc}(F)$ is a smallest $\mathcal{F}$\hy deletion set. There is
  fixed\hy parameter tractable algorithm that, given a graph $G$ and a
  parameter $k$, computes a feedback vertex set $D$ of $G$ such that $|D| \leq
  k$ or reports that no such set exists~\cite{ChenFominLiuLuVillanger07}. It
  follows that there is a fixed\hy parameter tractable algorithm, 
parameterized by the $\mathcal{F}$\hy deletion distance, for computing a 
  smallest $\mathcal{F}$\hy deletion set of an input CNF. Moreover, the
  incidence graphs of formulas in $\mathcal{F}$ have treewidth $1$, so
  $\mathcal{F}$ has the constructive few subterms property by
  Lemma~\ref{lemma:twfewsubterms}. Clearly, $\mathcal{F}$ is  closed under variable and clause deletion.  
Hence, applying Theorem~\ref{th:fptcompil}, we 
  conclude that CNFs have fixed-parameter tractable time computable OBDDs 
parameterized by $\mathcal{F}$\hy deletion distance.
\end{proof}

\section{Polynomial Size Incompilability}\label{sect:negative}
 
In this section, we introduce the \emph{subfunction width} of a graph CNF, 
to which the OBDD size of the graph CNF is exponentially related (Section~\ref{sect:manysubf}), 
and prove that \emph{expander graphs} yield classes of graph CNFs of \emph{bounded degree} 
with linear subfunction width, thus obtaining an exponential lower bound on the OBDD size 
for graph CNFs in such classes (Section~\ref{sect:degree}).

\subsection{Many Subfunctions}\label{sect:manysubf}

In this section, we introduce the \emph{subfunction width} of a graph CNF (Definition~\ref{def:spatious}), 
and prove that the OBDD size of a graph CNF is bounded below by an exponential function of its subfunction width (Theorem~\ref{th:epsilonspatious}).  

\medskip \noindent A \emph{graph CNF} is a CNF $F$ such that $F = \{ \{u,v\} \mid uv \in E \}$ for some graph $G = (V, E)$ without isolated vertices. 

\begin{definition}[Subfunction Width]\label{def:spatious}
Let $F$ be a graph CNF.  Let $\sigma$ be an ordering of $\mathsf{var}(F)$ and let $\pi$ be a prefix of $\sigma$.  
We say that a subset $\{c_1,\ldots,c_e\}$ of clauses in $F$ 
is \emph{subfunction productive} relative to $\sigma$ and $\pi$ if 
there exist $\{a_1,\ldots,a_e\} \subseteq \mathsf{var}(\pi)$ and 
$\{u_1,\ldots,u_e\} \subseteq \mathsf{var}(F)\setminus \mathsf{var}(\pi)$ such that for all $i,j \in \{1,\ldots,e\}$, $i \neq j$, and all $c \in F$, 
\begin{itemize}
\item $c_i=\{a_i,u_i\}$; 
\item $c \neq \{a_i,a_j\}$ and $c \neq \{a_i,u_j\}$.
\end{itemize}
The \emph{subfunction width} of $F$, in symbols $\mathsf{sfw}(F)$, 
is defined by 
$$\mathsf{sfw}(F)=\min_{\sigma} \max_{\pi} \{ |M| \mid \text{$M$ is subfunction productive relative to $\sigma$ and $\pi$} \}\text{,}$$
where $\sigma$ ranges over all orderings of $\mathsf{var}(F)$ 
and $\pi$ ranges over all prefixes of $\sigma$.  
\end{definition}

Intuitively, in the graph $G$ underlying the graph CNF $F$ in Definition~\ref{def:spatious}, 
there is a matching of the form $a_i u_i$ with $a_i \in \mathsf{var}(\pi)$ and $u_i \in \mathsf{var}(F) \setminus \mathsf{var}(\pi)$, $i \in \{1,\ldots,e\}$; 
such a matching is \lq\lq almost\rq\rq\ induced, in that $G$ can contain edges of the form $u_i u_j$, 
but no edges of the form $a_i a_j$ or $a_i u_j$, $i,j \in \{1,\ldots,e\}$, $i \neq j$.

\begin{theorem}\label{th:epsilonspatious}
Let $F$ be a graph CNF.  The OBDD size of $F$ is at least $2^{\mathsf{sfw}(F)}$.
\end{theorem}
\begin{proof}
Let $F$ be a graph CNF.  Let $D$ be any OBDD computing $F$, 
let $\sigma$ be the ordering of $\mathsf{var}(F)$ respected by $D$, 
and let $\pi$ be a prefix of $\sigma$ such that 
$\{c_1,\ldots,c_{e}\} \subseteq F$ is subfunction productive relative to $\sigma$ and $\pi$ 
and $e \geq \mathsf{sfw}(F)$.  Let $\{a_1,\ldots,a_e\} \subseteq \mathsf{var}(\pi)$ and $\{u_1,\ldots,u_e\} \subseteq \mathsf{var}(F)\setminus \mathsf{var}(\pi)$ 
be as in Definition~\ref{def:spatious}, so that in particular $c_i=\{a_i,u_i\}$, $i \in \{1,\ldots,e\}$.  Let
\begin{equation}\label{eq:foolingset}
L=\{ f \colon \mathsf{var}(\pi) \to \{0,1\} \mid \text{$f(v)=1$ for all $v \not\in \{a_1,\ldots,a_e\}$} \}\text{;}
\end{equation}
in words, $L$ is the set containing, for each assignment of $\{a_1,\ldots,a_e\}$, 
its extension to $\mathsf{var}(\pi)$ that sends all variables in $\mathsf{var}(\pi)\setminus \{a_1,\ldots,a_e\}$ to $1$.  

\begin{claim}
Let $f \in L$ and let $c \in F$ be such that $c \subseteq \mathsf{var}(\pi)$.  
Then, $f$ satisfies $c$.
\begin{proof}[Proof of Claim]
Otherwise, since $F$ is a graph CNF and by (\ref{eq:foolingset}) the only variables sent to $0$ by $f$ are in $\{a_1,\ldots,a_e\}$, 
it is the case that $c=\{a_i,a_j\}$ for some $i,j \in \{1,\ldots,e\}$, $i \neq j$, 
which is impossible by the second item in Definition~\ref{def:spatious}.
\end{proof}
\end{claim}

\begin{claim}
Let $f$ and $g$ be distinct assignments in $L$.  Then, $f$ and $g$ lead to different nodes in $D$.
\begin{proof}[Proof of Claim]
Let $f$ and $g$ be distinct assignments in $L$.  By the previous claim, 
$f$ and $g$ satisfy each clause in $F$ whose variables are contained 
in $\mathsf{var}(\pi)$.  Thus, the computation paths activated by $f$ and $g$ in $D$ 
lead to some nodes in $D$ distinct from the $0$-sink of $D$.  

Since $f$ and $g$ are distinct assignments in $L$, 
they differ on at least one variable in $\{a_1,\ldots,a_e\}$; 
say without loss of generality that $f(a_1)=0 \neq 1=g(a_1)$.  Let $h \colon \mathsf{var}(F)\setminus \mathsf{var}(\pi) \to \{0,1\}$ 
be such that $h(v)=0$ if and only if $v=u_1$.  We show that 
that $f \cup h$ does not satisfy $F$, but $g \cup h$ satisfies $F$; 
it follows that $f$ and $g$ lead to different nodes in $D$.

Clearly, $f \cup h$ does not satisfy $F$, 
because by Definition~\ref{def:spatious} the clause $c_1=\{a_1,u_1\}$ is in $F$, 
and by construction $f(a_1)=h(u_1)=0$.  We show that $g \cup h$ satisfies $F$.   

Let $c \in F$.  If $c \subseteq \mathsf{var}(\pi)$, 
then $g$ satisfies $c$ by the previous claim.  
If $c \subseteq \mathsf{var}(F)\setminus \mathsf{var}(\pi)$, 
then $h$ satisfies $c$, because $c$ contains two distinct variables, 
hence at least one of its variables differs from $u_1$ and is assigned to $1$ by $h$.  
Otherwise, $c \cap \mathsf{var}(\pi) \neq \emptyset$ 
and $c \cap (\mathsf{var}(F)\setminus \mathsf{var}(\pi)) \neq \emptyset$.  If $c$ 
contains a variable in $\mathsf{var}(F)\setminus \mathsf{var}(\pi)$ distinct from $u_1$, 
then $h$ satisfies $c$.  Otherwise, $c=\{a,u_1\}$ for some $a \in \mathsf{var}(\pi)$.  
In this case, if $a \in \{a_1,\ldots,a_e\}$, 
then $a=a_1$ by Definition~\ref{def:spatious}, 
and $g$ satisfies $c$ via $g(a_1)=1$.  Else, 
$a \in \mathsf{var}(\pi) \setminus \{a_1,\ldots,a_e\}$ 
and by definition of $L$ we have $g(a)=1$, 
so that again $g$ satisfies $c$.
\end{proof}
\end{claim}

It is readily observed that $|L|=2^e$.  Then, by the above claims, 
the computation paths activated by the assignments in $L$ 
lead to $2^e$ different nodes in $D$.  We observed that $e \geq \mathsf{sfw}(F)$.  
Then $D$ has size at least $2^{\mathsf{sfw}(F)}$.  It follows that the OBDD size 
of $F$ is at least $2^{\mathsf{sfw}(F)}$.  
\end{proof}

\subsection{Bounded Degree}\label{sect:degree}

In this section, we use the existence of a family of \emph{expander graphs} 
to obtain a class of graph CNFs with linear subfunction width (Lemma~\ref{lemma:lowerbound}), 
thus obtaining an exponential lower bound on the OBDD size of a class of CNFs of \emph{bounded degree} (Theorem~\ref{thm:mainnontechnical}).

\medskip \noindent Let $n$ and $d$ be positive integers, $d \geq 3$, 
and let $c<1$ be a positive real.  A graph $G=(V,E)$ is a $(n,d,c)$-\emph{expander} if 
$G$ has $n$ vertices, 
degree at most $d$, 
and for all subsets $W \subseteq V$ such that $|W| \leq n/2$, 
the inequality 
\begin{equation}\label{eq:expansion}
|\mathsf{neigh}(W)| \geq c|W|\text{.}
\end{equation}
It is known that for all integers $d \geq 3$, 
there exists a real $0<c$, and a sequence 
\begin{equation}\label{eq:expanders}
\{ G_i \mid i \in \mathbb{N} \} 
\end{equation}
such that $G_i=(V_i,E_i)$ is an $(n_i,d,c)$-expander ($i \in \mathbb{N}$), 
and $n_i$ tends to infinity as $i$ tends to infinity \cite[Section 9.2]{AS00}. 

\begin{lemma}\label{lemma:lowerbound}
Let $F$ be a graph CNF whose underlying graph is a $(n,a,d)$-expander ($n \geq 2$, $a>0$, $d \geq 3$).  
Then, $$\mathsf{sfw}(F) \geq \frac{\min\{1,a\}}{8d} \cdot n\text{.}$$
\begin{proof}
Let $F$ be a graph CNF whose underlying graph $G$ is a $(n,a,d)$-expander ($n \geq 2$, $a>0$, $d \geq 3$).  
Let $\sigma$ be any ordering of $\mathsf{var}(F)$, and let $\pi$ be the length $\lfloor n/2 \rfloor$ prefix of $\sigma$.  

\begin{claim}
There exists a subset $\{c_1,\ldots,c_e\}$ of clauses in $F$, subfunction productive relative to $\sigma$ and $\pi$, 
such that $e \geq \frac{\min\{1,a\}}{8d} \cdot n$.  
\begin{proof}[Proof of Claim]
We construct size $e$ sets $\{c_1,\ldots,c_e\} \subseteq F$, 
$\{a_1,\ldots,a_e\} \subseteq \mathsf{var}(F)$, 
and $\{u_1,\ldots,u_e\} \subseteq \mathsf{var}(F)$ 
by iterating the following ($j=1,\ldots,e$):
\begin{itemize}
\item Pick an edge $w_j w_j' \in F$ between $w_j \in \mathsf{var}(\pi)$ and $w_j' \in \mathsf{neigh}(\mathsf{var}(\pi))$.
\item Settle $a_j=w_j$, $u_j=w_j'$, and $c_j=\{a_j,u_j\}$.
\item Delete $\mathsf{neigh}(w_j)$ and $\mathsf{neigh}(w_j') \cap \mathsf{var}(\pi)$.
\end{itemize}
Clearly, 
$\{a_1,\ldots,a_e\} \subseteq \mathsf{var}(\pi)$ 
and $\{u_1,\ldots,u_e\} \subseteq \mathsf{var}(F)\setminus \mathsf{var}(\pi)$.

Each iteration deletes at most $2d$ vertices in $\mathsf{var}(\pi)$ 
(the neighbors of $w_j$ in $\mathsf{var}(\pi)$, at most $d$ vertices, 
and the neighbors of $w_j'$ in $\mathsf{var}(\pi)$, at most $d$ vertices, including $w_{j}$), 
and at most $d$ vertices in $\mathsf{neigh}(\mathsf{var}(\pi))$ 
(the neighbors of $w_j$ in $\mathsf{neigh}(\mathsf{var}(\pi))$, including $w'_j$).  Since 
$|\mathsf{var}(\pi)|=\lfloor n/2 \rfloor$ 
and $|\mathsf{neigh}(\mathsf{var}(\pi))| \geq a \lfloor n/2\rfloor$ by (\ref{eq:expansion}), 
the number of steps is 
$$e \geq \min\left\{\frac{\lfloor n/2\rfloor}{2d},\frac{a \lfloor n/2\rfloor}{d}\right\} \geq 
\frac{\min\{1,a\}}{2d} \Big\lfloor \frac{n}{2}\Big\rfloor \geq 
\frac{\min\{1,a\}}{4d} \cdot (n-1) \geq
\frac{\min\{1,a\}}{8d} \cdot n
\text{,}$$
since $n \geq 2$.  

We now check Definition~\ref{def:spatious}.  By construction, $c_j=\{a_j,u_j\}$ for all $j \in \{1,\ldots,e\}$.  
Moreover, let $j,j' \in \{1,\ldots,e\}$, $j \neq j'$, and let $c \in F$.  
Say without loss of generality that $j<j'$.  Assume that $c=\{a_j,a_{j'}\}$.  
Then there exist $w_j=a_j \in \mathsf{var}(\pi)$ at step $j$, 
and $w_{j'}=a_{j'} \in \mathsf{var}(\pi)$ at step $j'$, 
such that $w_j w_{j'} \in F$.  
Then, $w_{j'} \in \mathsf{neigh}(w_j)$, 
so that it is deleted at step $j$; 
but $w_{j'}$ exists at step $j'>j$, a contradiction.  Finally 
assume that $c=\{a_j,u_{j'}\}$ or $c=\{a_{j'},u_j\}$.  
If $c=\{a_j,u_{j'}\}$, then there exist $w_j=a_j \in \mathsf{var}(\pi)$ at step $j$, 
and $w'_{j'}=u_{j'} \in \mathsf{neigh}(\mathsf{var}(\pi))$ at step $j'$, 
such that $w_j w'_{j'} \in F$.  
Then, $w'_{j'} \in \mathsf{neigh}(w_j)$ is deleted at step $j$, 
but it exists at step $j'>j$, a contradiction.  
If $c=\{a_{j'},u_j\}$, then there exist $w'_j=u_j \in \mathsf{neigh}(\mathsf{var}(\pi))$ at step $j$, 
and $w_{j'}=a_{j'} \in \mathsf{var}(\pi)$ at step $j'$, 
such that $w_{j'} w'_{j} \in F$.  
Then, $w_{j'} \in \mathsf{neigh}(w'_j) \cap \mathsf{var}(\pi)$ is deleted at step $j$, 
but it exists at step $j'>j$, a contradiction.  
\end{proof}
\end{claim}

The claim implies that $\mathsf{sfw}(F) \geq \frac{\min\{1,a\}}{8d} \cdot n$.  
\end{proof}
\end{lemma}

\begin{theorem}\label{thm:mainnontechnical}
There exist a class $\mathcal{F}$ of CNF formulas and a constant $c>0$ such that, 
for every $F \in \mathcal{F}$, the OBDD size of $F$ is at least $2^{c \cdot \mathsf{size}(F)}$.  
In fact, $\mathcal{F}$ is a class of read $3$ times, monotone, $2$-CNF formulas.
\begin{proof}
Let $\mathcal{G}=\{ G_i \mid i \in \mathbb{N} \}$ be a family of graphs as in (\ref{eq:expanders}), 
so that for all $i \in \mathbb{N}$ the graph $G_i=(V_i,E_i)$ is a $(n_i,d,a)$-expander ($n_i \geq 2$, $d=3$, $a>0$) 
and $n_i \to \infty$ as $i \to \infty$.  Note that, using the expansion property, it is readily verified that 
each graph in $\mathcal{G}$ is connected; in particular, it does not have isolated vertices.  
Therefore $\mathcal{F}=\{ E_i \colon i \in \mathbb{N} \}$ is a class of graph CNFs; 
indeed, it is a class of read $3$ times, monotone, $2$-CNF formulas.  

Let $F \in \mathcal{F}$.  By Lemma~\ref{lemma:lowerbound}, we have that $$\mathsf{sfw}(F) \geq \frac{\min\{1,a\}}{8d} \cdot |\mathsf{var}(F)|\text{.}$$
Since the underlying graph of $F$ has degree at most $d$ and $|\mathsf{var}(F)|$ vertices, 
the $F$ contains at most $d|\mathsf{var}(F)|$ clauses (each variable occurs in at most $d$ clauses), 
and each clause contains at most $2$ literals.  Therefore, $2d |\mathsf{var}(F)| \geq \size{F}$.  Thus, 
$$\mathsf{sfw}(F) \geq \frac{\min\{1,a\}}{8d} \cdot |\mathsf{var}(F)| \geq \frac{\min\{1,a\}}{16 d^2} \cdot 2d|\mathsf{var}(F)| \geq \frac{\min\{1,a\}}{16 d^2} \cdot \size{F}\text{.}$$  

It follows from Theorem~\ref{th:epsilonspatious} that the OBDD size of $F$ is at least $2^{c \cdot |\mathsf{var}(F)|}$ 
where $c=\min\{1,a\}/ 16 d^2$, and we are done.
\end{proof}
\end{theorem}

\section{Conclusion}\label{sect:conclusion}

We have proved new lower and upper bound results on the OBDD size 
of structurally characterized CNF classes, pushing the frontier significantly 
beyond the current knowledge, as depicted in Figure~\ref{fig:hierarchy}.  
We conclude mentioning that tightening the gap left by this work in the considered hierarchy of 
structural CNF classes seems to require new ideas.  

As far as upper bounds are concerned, the few subterms property 
is a natural source of polynomial upper bounds; for instance, 
the width measure recently introduced by Oztok and Darwiche in the compilation of CNFs into DNNFs 
(a more general formalism than OBDDs), 
once instantiated to OBDDs, is closely related to our subterm width measure \cite{OD14}.  However, 
the frontier charted in this work seems to push the few subterms property to its limits, 
in the sense that natural variable orderings do not yield the few subterms property for classes lying immediately beyond the frontier, 
namely (clause) convex CNFs and bounded clique\hy width CNF classes.

As for lower bounds, the technique based on expander graphs essentially requires bounded degree, 
but the candidate classes for improving lower bounds in our hierarchy, bounded clique\hy width CNFs and beta acyclic CNFs,
have unbounded degree. In fact, in both cases, imposing a degree bound leads to classes
of bounded treewidth~\cite{KLM09} and thus polynomial bounds on the size of OBDD representations.


\begin{thebibliography}{1}
\bibitem{AS00} 
N. Alon and J. Spencer.
\newblock {\em The Probabilistic Method}.
\newblock Wiley, 2000.

\bibitem{Bodlaender1996Arboretum} 
H. Bodlaender.
\newblock A Partial $k$-Arboretum of Graphs with Bounded Treewidth.
\newblock {\em Theor. Comput. Sci.}, 209(1--2): 1--45, 1998.

\bibitem{BodlaenderKloks96}
H. Bodlaender and T. Kloks.
\newblock Efficient and Constructive Algorithms for the Pathwidth and
Treewidth of Graphs
\newblock {\em J. Algorithms}, 21(2): 358--402, 1996.


\bibitem{BW99} 
B. Bollig and I. Wegener.
\newblock Complexity Theoretical Results on Partitioned (Nondeterministic) Binary Decision Diagrams.
\newblock {\em Theory Comput. Syst.}, 32(4): 487--503, 1999.

\bibitem{BL76} 
K. Booth and G. Lueker.
\newblock Testing for the Consecutive Ones Property, Interval Graphs, and Graph Planarity Using PQ-Tree Algorithms.
\newblock {\em J. Comput. Syst. Sci.}, 13(3):335--379, 1976.

\bibitem{BLS99} 
A. Brandst{\"a}dt, V. Le, and J. Spinrad.
\newblock {\em Graph Classes: a Survey}.
\newblock SIAM, 1999.

\bibitem{BBCM14} 
J. Brault-Baron, F. Capelli, and Stefan Mengel.
\newblock Understanding Model Counting for $\beta$-Acyclic CNF-Formulas.
\newblock Preprint in {\em CoRR}, abs/1405.6043, 2014.

\bibitem{ChenFominLiuLuVillanger07} 
J. Chen, F. Fomin, Y. Liu, S. Lu, and Y. Villanger.
\newblock Improved Algorithms for the Feedback Vertex Set Problems.
\newblock In {\em WADS}, 2007.


\bibitem{CO00} 
B. Courcelle and S. Olariu.
\newblock Upper Bounds to the Clique Width of Graphs.
\newblock {\em Discrete Appl. Math.}, 101(1--3): 77--144, 2000.

\bibitem{D93} 
S. Devadas.
\newblock Comparing Two-Level and Ordered Binary Decision Diagram Representations of Logic Functions.
\newblock {\em IEEE T. Comput. Aid. D.}, 12(5): 722--723, 1993.

\bibitem{D2010}
R. Diestel.
\newblock {\em Graph Theory}.
\newblock Springer, 2010.

\bibitem{DowneyFellows13} 
R. Downey and M. Fellows.
\newblock {\em Fundamentals of Parameterized Complexity}.
\newblock Springer, 2013.

\bibitem{FPV05} 
A. Ferrara, G. Pan, and M. Vardi. 
\newblock Treewidth in Verification: Local vs Global. 
\newblock In {\em LPAR}, 2005.

\bibitem{FlumGrohe} 
J. Flum and M. Grohe.
\newblock {\em Parameterized Complexity Theory}.
\newblock Springer, 2006.

\bibitem{FFK88} 
M. Fujita, H. Fujisawa, and N. Kawato.
\newblock Evaluation and Improvements of Boolean Comparison Method Based on Binary Decision Diagrams.
\newblock In {\em ICCAD}, 1988.


\bibitem{Gallo84}
G. Gallo.
\newblock An O(n log n) algorithm for the convex bipartite matching problem.
\newblock{\em Oper. Res. Lett.}, 3(1): 31--34, 1984.

\bibitem{Glover67}
F. Glover.
\newblock Maximum matching in a convex bipartite graph.
\newblock{\em Naval Research Logistics Quarterly.}, 14(3): 313--316, 1967.

\bibitem{HLW06} 
S. Hoory, N. Linial, and A. Wigderson.
\newblock Expander Graphs and their Applications. 
\newblock {\em Bull. Amer. Math. Soc.}, 43:439--561, 2006.

\bibitem{J12} 
S. Junkna.
\newblock {\em Boolean Function Complexity}.
\newblock Springer, 2012.

\bibitem{KLM09} 
M. Kaminski, V. Lozin, and M. Milanic.
\newblock Recent Developments on Graphs of Bounded Clique-Width.
\newblock {\em Discrete Appl. Math.}, 157(12): 2747--2761, 2011.

\bibitem{KKLV11} 
J. K{\"o}bler, S. Kuhnert, B. Laubner, and O. Verbitsky. 
\newblock Interval Graphs: Canonical Representations in Logspace.
\newblock {\em SIAM J. Comput.}, 40(5): 1292--1315, 2011.

\bibitem{LR04} 
V. Lozin and D. Rautenbach. 
\newblock Chordal Bipartite Graphs of Bounded Tree- and Clique-Width.
\newblock {\em Discrete Math.}, 283: 151--158, 2004.

\bibitem{L11} 
A. Lubotzky.
\newblock Expander Graphs in Pure and Applied Mathematics.
\newblock Preprint in {\em CoRR}, abs/1105.2389, 2011.

\bibitem{OD14} 
U. Oztok and A. Darwiche.
\newblock CV-Width: A New Complexity Parameter for CNFs.
\newblock In {\em ECAI}, 2014.

\bibitem{R14} 
I. Razgon.
\newblock On OBDDs for CNFs of Bounded Treewidth.
\newblock In {\em KR}, 2014.

\bibitem{STV14} 
S. S{\ae}ther, J. Telle, and M. Vatshelle.
\newblock Solving MaxSAT and \#SAT on Structured CNF Formulas.
\newblock In {\em SAT}, 2014.

\bibitem{SY93}
G. Steiner and S. Yeomans.
\newblock Level Schedules for Mixed-model, Just-in-Time Processes.
\newblock In {\em Manage. Sci.}, 39: 728--735, 1993.

\bibitem{W00} 
I. Wegener.
\newblock {\em Branching Programs and Binary Decision Diagrams}.
\newblock SIAM, 2000.

\end{thebibliography}
\end{document}